\documentclass[journal]{IEEEtran}

\usepackage{amsmath,amssymb,dsfont,graphicx,bm,pstricks,subfigure}
\usepackage{enumerate}
\usepackage{calc, cite}

\newtheorem{theorem}{Theorem}{}

\newtheorem{proposition}[theorem]{Proposition}

\newtheorem{definition}[theorem]{Definition}

\newenvironment{example}[1][Example.]{\begin{trivlist}
\item[\hskip \labelsep {\bfseries #1}]}{\end{trivlist}}

\renewcommand{\vec}[1]{\ensuremath{\MakeLowercase{\boldsymbol{#1}}}}

\providecommand{\vc}{\vec{c}}

\providecommand{\vx}{\vec{x}} \providecommand{\vy}{\vec{y}}

\newcommand{\cM}{\boldsymbol{{M}}}
\newcommand{\cMc}{\cM^{(\mathsf{c})}}
\newcommand{\cMs}{\cM^{(\mathsf{s})}}
\newcommand{\F}{\mathds{F}}
\newcommand{\supp}{\mathsf{supp}}

\title{Group Testing with Probabilistic Tests: Theory, Design and Application} \author{Mahdi Cheraghchi,
 \IEEEmembership{Member, IEEE},
 Ali Hormati,
  \IEEEmembership{Member, IEEE},
   Amin Karbasi,
   \IEEEmembership{Student Member, IEEE},
  and Martin Vetterli
  \IEEEmembership{Fellow, IEEE}

  \thanks{{\footnotesize  
  M.~Cheraghchi is with the Department of Computer Science, University of Texas at Austin, USA (email: mahdi@cs.utexas.edu).
  
  A.~Hormati, A.~Karbasi, and M.~Vetterli are with the School of Computer and Communication
      Sciences, EPFL, Switzerland (emails: \{ali.hormati, amin.karbasi, martin.vetterli\}@epfl.ch).
      M.~Vetterli is also with the Department of Electrical Engineering and Computer Sciences, University of California at Berkeley, USA.
      
      A preliminary summary of this work appeared in proceedings of the 47th  Allerton conference on communication, control, and computing~\cite{CheraghchiHKV09}.} Part of the work was done while M.~Cheraghchi was with the School of Computer and Communication Sciences at EPFL (supported by SNSF grant 200020-115983/1).
      The remaining authors were supported by ERC grant 247006 and SNSF grants 200021-121935 and 5005-67322.}  }

\date{}

\begin{document}

\maketitle

\begin{abstract}

\noindent Identification of defective members of large populations has been widely studied in the statistics community under the name of group testing. It involves grouping subsets of items into different pools and detecting defective members based on the set of test results obtained for each pool.

In a classical noiseless group testing setup, it is assumed that the sampling procedure is fully known to the reconstruction algorithm, in the sense that the existence of a defective member in a pool results in the test outcome of that pool to be positive. However, this may not be always a valid assumption in some cases of interest. In particular, we consider the case where the defective items in a pool can become independently \emph{inactive} with a certain probability. Hence, one may obtain a negative test result in a pool despite containing some defective items. As a result, any sampling and reconstruction method should be able to cope with two different types of uncertainty, i.e., the unknown set of defective items and the partially unknown, probabilistic testing procedure.

In this work, motivated by the application of detecting infected people in viral epidemics, we design non-adaptive sampling procedures that allow successful identification of the defective items through a set of probabilistic tests. Our design requires only a small number of tests to single out the defective items. In particular, for a population of size $N$ and at most $K$ defective items with activation probability $p$, our results show that $M = O(K^2\log{(N/K)}/p^3)$ tests is sufficient if the sampling procedure should work for \emph{all} possible sets of defective items, while $M = O(K\log{(N)}/p^3)$ tests is enough to be successful for any single set of defective items. Moreover, we show that the defective members can be recovered using a simple reconstruction algorithm with complexity of $O(MN)$.
\end{abstract}
\begin{IEEEkeywords}
Group testing, probabilistic tests, sparsity recovery, compressed sensing, epidemiology.
\end{IEEEkeywords}

\begin{section}{Introduction}
\IEEEPARstart{I}{nverse} problems, with the goal of recovering a signal from partial and noisy observations, come in many different formulations and arise in many  applications. One important property of an inverse problem is to be well-posed, i.e., there should exist a unique and stable solution to the problem~\cite{Hadamard23}. In this regard, prior information about the solution, like sparsity, can be used as a ``regularizer'' to transform an ill-posed problem to a well-posed one. In this work, we look at a particular inverse problem with less measurements than the number of unknowns (ill-posed) but with sparsity constraints on the solution. As will be explained in detail, the interesting aspect of this problem is that the sampling procedure is probabilistic and not fully known at recovery time.

Suppose that in a large set of items of size $N$, at most $K \ll N$ of them are defective and we wish to identify this small set of defective items. By testing each member of the set separately, we can expect the cost of the testing procedure to be large. If we could instead pool a number of items together and test the pool collectively, the number of tests required might be reduced. The ultimate goal is to construct a pooling design to identify the defective items while minimizing the number of tests.  This is the main conceptual idea behind the classical \emph{group testing} problem which was introduced by Dorfman~\cite{Dorfman43} and later found applications in a variety of areas. The first important application of the idea dates back to World War~II when it was suggested for syphilis screening. A few other examples of group testing applications include testing for defective items (e.g., defective light bulbs or resistors) as a part of industrial quality assurance~\cite{SobelG59}, DNA sequencing~\cite{PevznerL94}, DNA library screening in molecular biology (see, e.g.,~\cite{NgoD00, SchliepTR03,Macula99,Macula99b,ChengD08} and the references therein), multi-access communication~\cite{Wolf85}, data compression~\cite{HongL02}, pattern matching~\cite{CliffordEPR07}, streaming algorithms~\cite{CormodeM05}, software testing~\cite{BlassG02} and compressed sensing~\cite{CormodeM06}. See the books by Du and Hwang for a detailed account of the major developments in this area~\cite{DuH99,DuH06}.

In a classical group testing setup, it is assumed that the reconstruction algorithm has full access to the sampling procedure, i.e., it knows which items participate in each pool. Moreover, if the tests are reliable, the existence of a defective item in a pool results in a positive test outcome. In an unreliable setting, the test results can be contaminated by false positives and/or false negatives. Compared to the reliable setting, special care should be taken to tackle this uncertainty in order to successfully identify the defective items. However, in some cases of interest, there can exist other types of uncertainty that challenge the recovery procedure.

In this work, we investigate the group testing problem with probabilistic tests. In this setting, a defective item which participate in a pool can be \emph{inactive}, i.e., the test result of a pool can be negative despite containing some defective items. Therefore, a negative test result does not indicate that all the items in the corresponding pool are non-defective with certainty. We follow a probabilistic approach to model the activity of the defective items, i.e., each defective item is active independently in each pool with probability $p$. Therefore, the tests contain uncertainty not in the sense of false positives or false negatives, but in the sense of the underlying probabilistic testing procedure. More precisely, let us denote by $\cMc$ the designed contact matrix which indicates the items involved in each pool, i.e.,
\begin{equation*}
\cMc_{ij} = \left\{\begin{array}{l l}
     1 & \text{if test $i$ includes item $j$} \\
     0 & \text{otherwise.} \\
   \end{array}\right.
\end{equation*}
The probabilistic tests are then given by
\begin{equation*}
  \vy = \cMs \cdot \vx
\end{equation*}
where $\cMs$ denotes the probabilistic sampling matrix, $\vx$ is the sparse input vector and $\vy$ denotes the vector of test results. Each element of the contact matrix $\cMc$ is independently mapped to the corresponding element of the sampling matrix $\cMs$ by passing through the channel shown in Figure~\ref{fig:Zchannel}~\cite{AtiaS09}. In fact, the zeros in the contact matrix remain zeros in the sampling matrix while the ones are mapped to zeros with probability $1-p$.

\begin{figure}[t]
\center
\includegraphics[width=\columnwidth]{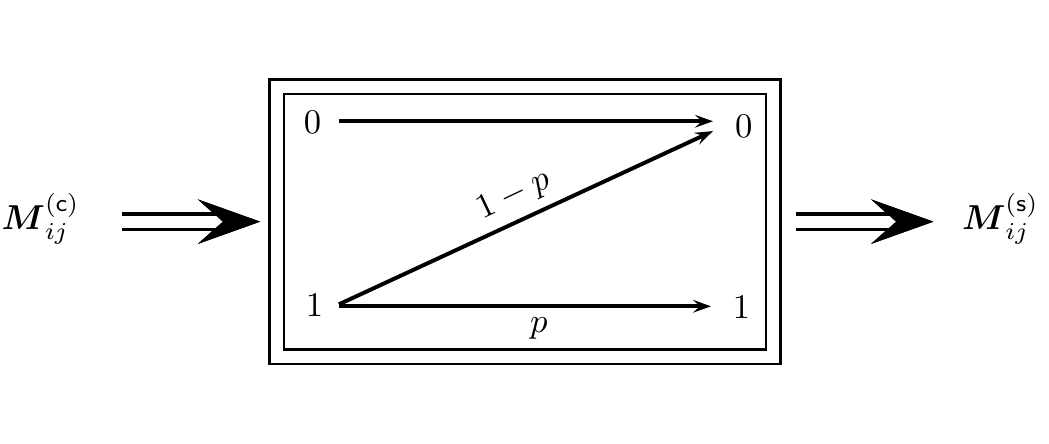}
\caption{Each element of the sampling matrix $\cMs_{ij}$ is generated independently from the corresponding element of the contact matrix $\cMc_{ij}$ by passing through a channel. The zeros in the contact matrix remain zeros in the sampling matrix while the ones are converted to zeros with probability $1-p$. }\label{fig:Zchannel}
\end{figure}

In this work, our goal is to design efficient sampling and recovery mechanisms to successfully identify the sparse vector $\vx$, despite the partially unknown testing procedure given by the sampling matrix $\cMs$. Our interest is in non-adaptive sampling procedures in which the sampling strategy (i.e., the contact matrix) is designed before seeing the test outcomes. In our analysis, we consider two different design strategies: In the \emph{per-instance} design,  the sampling procedure should be suitable for a fixed set of defective items with overwhelming probability while in the \emph{universal} design, it should be appropriate for all possible sets of defective items. We show that $M = O\left(K \log (N) / p^3\right)$ tests are sufficient for successful recovery in the per-instance scenario while we need $M = O\left(K^2 \log (N/K) / p^3\right)$ tests for the universal design. Moreover, the defective items can be recovered by a simple recovery algorithm with complexity of $O(MN)$.
For a constant parameter $p$, the bounds on the number of measurements are asymptotically tight up to logarithmic factors. This is simply because standard group testing that corresponds to the case $p=1$ requires $M=\Omega(K^2 \log_K(N))$ non-adaptive measurements in the universal setting (cf.\ \cite[Ch~7]{DuH99}) and $M=\Omega(K \log(N/K))$ measurements in the per-instance scenario (by a ``counting argument'').

\begin{figure}[t]
\centering
\includegraphics[width=0.7\columnwidth]{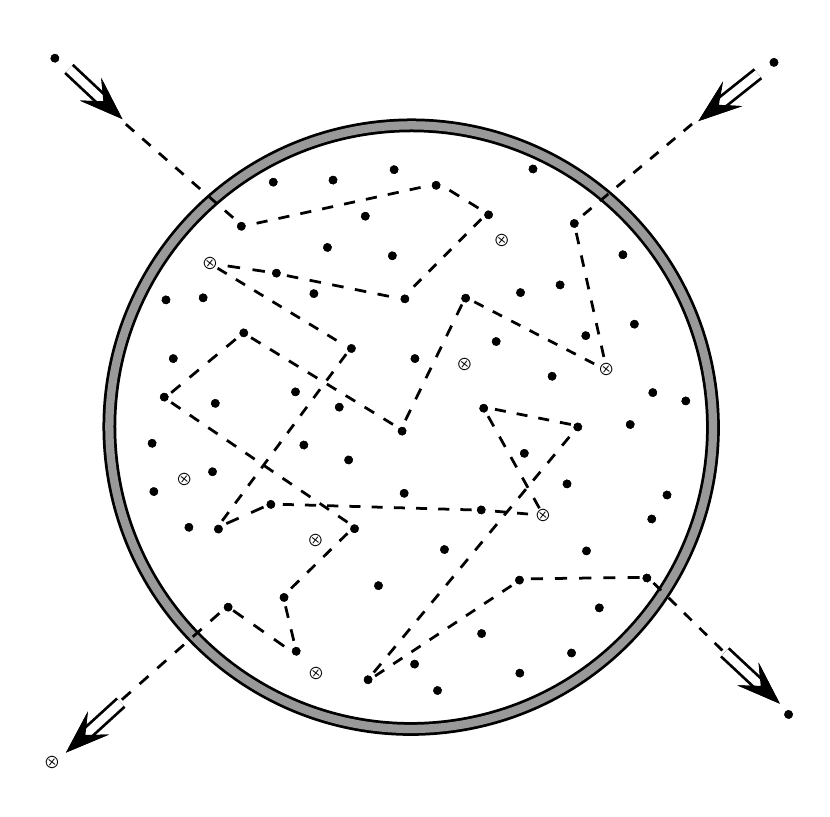}
\caption[ ]{Collective sampling using agents in viral epidemics. ($\bullet$) symbols represent healthy people while ($\otimes$) symbols indicate infected ones. The dashed lines connect the individuals contacted by the agents. An agent may remain healthy despite having contact with some infected people.}\label{chap6fig:epidemiologySetup}
\end{figure}

The above-mentioned probabilistic sampling procedure can well model the sampling process in an epidemiology application, where the goal is to successfully identify a sparse set of virally-infected people in a large population with a few collective tests. In viral epidemics, one way to acquire collective samples is by sending ``agents'' inside the population whose task is to contact people. Once an agent makes contact with an ``infected'' person, there is a \textit{chance} that he gets infected, too. By the end of the testing procedure, all agents are gathered and tested for the disease. Note that, when an agent contacts an infected person, he will get infected with a certain probability typically smaller than one. Hence, it may well happen that an agent's result is negative (meaning that he is not infected) despite a contact with some infected person. One can assume that each agent has a \emph{log file} by which one can figure out with whom he has made contact. One way to implement the log in practice is to use identifiable devices (for instance, cell phones) that can exchange unique identifiers when in range. This way, one can for instance ask an agent to randomly meet a certain number of people in the population and at the end, learn which individuals have been met from the data gathered by the device. However, one should assume that when an agent gets infected, the resulting infection will not be contagious, i.e., an agent never infects other people\footnote{This assumption is reasonable with certain diseases when there is an incubation time.}.
In the above model, the agents can in fact take many forms, including people who happen to be in contact with random individuals within the population (e.g., cashiers, bus drivers, etc.). 

The model explained using the epidemiology example above can in fact capture a broader range of settings, and in particular, any group testing problem where items can be defective with a certain probability. An example of such a setting is testing for faulty components (or modules) in digital logic systems (e.g., an integrated circuit).  This can be modeled through a probabilistic setting  where
the probability $p$ denotes the percentage of time that a faulty component
does not work correctly. In this way, one can use our group testing
results to efficiently localize the few unreliable circuitry elements
out of a large set of components on a given chip. One should note that
this model generalizes the classical application of group testing for
fault detection in electronic circuits where components are assumed to
be either fully reliable or fully unreliable (see \cite{ref:Kahng}).

The organization of this paper is as follows. We first give an overview of the related work in Section~\ref{chap6:relatedWork} which is then followed in Section~\ref{chap6:problemDefinition} by a more precise formulation of our problem. In order to solve the original stochastic problem, we first solve an adversarial variation of it in Section~\ref{chap6:adversarial} which we find more convenient to work with. Then, in Section~\ref{chap6sec:probDesign} and by using the results obtained from the adversarial setting, we design sensing and recovery procedures to efficiently solve the original stochastic problem. In Section~\ref{chap6:design_simulation}, we provide a systematic design procedure which provides us with the exact value for the number of tests, along with the other necessary parameters, as a function of the desired probability of unsuccessful reconstruction. We evaluate the design by doing a set of numerical experiments. The paper is summarized in Section~\ref{chap6:summary}.

\end{section}

\begin{section}{Related Work}
\label{chap6:relatedWork}
A large body of work in classical group testing has focused on combinatorial design, \textit{i.e.}, construction of contact matrices that satisfy a disjunctness property (the exact definition will be provided in Section~\ref{chap6:adversarial}). Matrices that have this property are of significant interest since they ensure identifiability of defective items and moreover, they lead to efficient recovery algorithms. This property has been extensively studied in \cite{Dyachkov82, Dyachkov83, Erdos82, Erdos85, Hwang87}. By using probabilistic methods, authors in \cite{Dyachkov82} developed upper and lower bounds on the number of tests/rows for the contact matrix to be $K$-disjunct. More precisely, they showed that the number of rows should at least scale asymptotically as $O(K^2 \log N/\log K)$ for exact reconstruction with worst case input. On the other hand, a randomly generated matrix will be $K$-disjunct with high probability if the number of rows scales as $O(K^2 \log (N/K))$ \cite{Dyachkov83}. Having a $K$-disjunct matrix, one can devise an efficient reconstruction algorithm to identify up to $K$ defective items. This is true if the reconstruction algorithm fully knows the sampling procedure. However, in our scenario, the decoder has to cope simultaneously with two sources of uncertainty, the unknown group of defective items and the partially unknown (or stochastic) sampling procedure. For this reason we need to use a more general form of disjunctness property.

We should also point out the relationship between our setup and the compressed sensing (CS) framework~\cite{CandesRT06, Donoho06}. In CS, a random projection of a  sparse signal is given and the goal is to find the position as well as the value of the non-zero entries of the input signal while keeping the number of projections to a minimum. Exploiting the similarity between group testing and CS, new recovery algorithms for sparse signal recovery have been proposed in \cite{Gilbert06} and \cite{Gilbert08}. Although in CS the goal is to measure and reconstruct sparse signals with few measurements, it differs in significant ways from our setup. In CS, it is typically assumed that the decoder knows the measurement matrix a priori\footnote{There are, however, works that consider compressed sensing under small perturbations of the measurement matrix (cf.\ \cite{ref:HSD10}).
A large body of the compressed sensing literature considers a noise model where the measurement outcomes are perturbed by a real-valued noise vector, while the measurement matrix is exact. See, for example, \cite{AkcakayaT08,WangWR10,ref:FRG09,ReevesG08} and the references therein.
}. However, this is not the case in our setup. In other words, by using the language of compressed sensing, the measurement matrix might be ``noisy'' and not precisely known to the decoder. As it turns out, by using a sufficient number of tests this issue can be resolved. Another difference is that in CS, the input and the measurement matrix are real valued and operations are performed on real numbers whereas in our case, the input vector, the measurement matrix and the arithmetic are all boolean.

Recently, the authors in~\cite{AtiaS09} investigated the probabilistic testing model that we consider in this paper from  an information theoretic perspective.  Unlike our combinatorial approach, they use information theoretic techniques to obtain bounds on the required number of tests. Namely, they get $M=O(K^2 (\log N) / p^{2})$ and $M=O(K (\log N) / p^{2})$ measurements for universal and per-instance scenarios, respectively, which is asymptotically comparable to what we obtain in this work (for a fixed $p$ and $K \ll N$). To achieve the bounds, they consider typical set decoding as the reconstruction method. However, they do not provide a practical, low complexity decoding algorithm for the reconstruction.

Another work that is relevant to ours is \cite{Hwang76} that considers group testing under a ``dilution'' effect. This model is targeted for biological experiments where combining items in a group may cause defected items go undetected when the size of the group is large. In particular, their model assumes that each item is independently defected with a certain (fixed) probability, and a defected item in a group of size $t$ affects the test with a probability proportional to $1/t$ (thus, a ``diluted'' group with few defectives becomes more likely to test negative as its size grows). They analyze the number of required tests using a simple (but sub-optimal) test design originally proposed by Dorfman~\cite{Dorfman43}.

\end{section}
\begin{section}{Problem Definition}
\label{chap6:problemDefinition}
To model the problem, we enumerate the population from $1$ to $N$ and the tests from $1$ to $M$. Let the nonzero entries of $\vx := (x_1, x_2, \ldots, x_N) \in \F_2^N$ indicate the defective items within the population, where $\F_2$ is the finite field of order $2$. Moreover, we assume that $\vx$ is a $K$-sparse vector, i.e., it has at most $K$ entries equal to one (corresponding to the defective items). We refer to the \emph{support set} of $\vx$, denoted by $\supp(\vx)$, as the set which contains positions of the nonzero entries.

As is typical in the literature of group testing, we introduce an $M\times N$ boolean \textit{contact} matrix $\cMc$ to model the set of non-adaptive tests. We set $\cMc_{ij}$ to one if and only if the $i$th test contains the $j$th item. The matrix $\cMc$ only shows which tests contain which items. In particular, it does not indicate whether the tests eventually get affected by the defective items. Let us assume that when a test contains a set of defective items, each of them makes the test positive independently with probability $p$, which is a fixed parameter that we call the \emph{activation probability}. Therefore, the real \textit{sampling} matrix $\cMs$ can be thought of as a variation of $\cMc$ in the following way: Each nonzero entry of $\cMc$ is flipped to $0$ independently with probability $1-p$\,. Then, the resulting matrix $\cMs$ is used just as in classical group testing to produce the test vector $\vy \in \F_2^M$
\begin{equation*}
  \vy=\cMs \cdot \vx
\end{equation*}
where the arithmetic is boolean, i.e., multiplication with the logical AND and addition with the logical OR.

The contact matrix $\cMc$, the test vector $\vy$, the upperbound on the number of nonzero entries $K$, and the activation probability $p$ are known to the decoder, whereas the sampling matrix $\cMs$ (under which the collective samples are taken) and the input vector $\vx$ are unknown. The task of the decoder is to identify the $K$ nonzero entries of $\vx$ based on the known parameters.

\begin{example}
As a toy example, consider a population with $6$ items where only two of them (items $3$ and $4$) are defective. We do a set of three tests, where the first one contains items $1,3,5$, the second one contains items $2,4,6$, and the third one contains items $2,3,5,6$. Therefore, the contact matrix and the input vector have the following form
\begin{eqnarray*}
 \vx&=&(\begin{array}{c c c c c c}  0 & 0 & 1 & 1 & 0 & 0\end{array})^T,\\
 \cMc&=&\left(\begin{array}{c c c c c c}
     1 & 0 & 1 & 0 & 1 & 0 \\
     0 & 1 & 0 & 1 & 0 & 1 \\
     0 & 1 & 1 & 0 & 1 & 1 \\
   \end{array}\right)\,.
\end{eqnarray*}
Let us assume that only the second test result is positive. This means that the test vector is
\begin{equation*}
  \vy=(\begin{array}{c c c}  0 & 1 & 0 \end{array})^T\,.
\end{equation*}
As we can observe, there are many possibilities for the sampling matrix, all of the following form:
\begin{eqnarray*}
  \cMs&=&\left(\begin{array}{c c c c c c}
     ? & 0 & ? & 0 & ? & 0 \\
     0 & ? & 0 & ? & 0 & ? \\
     0 & ? & ? & 0 & ? & ? \\
  \end{array}\right)
\end{eqnarray*}
where the question marks are $0$ with probability $1-p$ and $1$ with probability $p$. It is the decoder's task to figure out which combinations make sense based on the outcome vector. For example, the following matrices and input vectors fit perfectly with $\vy$:
\begin{eqnarray*}
  \left(\begin{array}{c}  0 \\ 1 \\0 \end{array}\right)&=&\left(\begin{array}{c c c c c c}
     1& 0 & 0 & 0 & 1 & 0 \\
     0 & 1 & 0 & 1 & 0 & 1 \\
     0 & 1 & 0 & 0 & 1 & 1 \\
  \end{array}\right) \left(\begin{array}{c}  0 \\ 0 \\ 1 \\ 1 \\ 0 \\ 0\end{array}\right)\\
  &=&\left(\begin{array}{c c c c c c}
     1& 0 & 1 & 0 & 1 & 0 \\
     0 & 1 & 0 & 1 & 0 & 1 \\
     0 & 1 & 1 & 0 & 1 & 0 \\
  \end{array}\right) \left(\begin{array}{c}  0 \\ 0 \\ 0 \\ 1 \\ 0 \\ 1\end{array}\right).
\end{eqnarray*}\hspace{\textwidth}
\end{example}
\vspace{2mm}
More formally, the goal of our scenario is two-fold:
\begin{enumerate}
\item Designing the contact matrix $\cMc$ so that it allows unique reconstruction of sparse input $\vx$ from outcome $\vy$ with overwhelming probability $1-o(1)$ over the randomness of the sampling matrix $\cMs$.
\item Proposing a recovery algorithm with low computational complexity.
\end{enumerate}

We present a probabilistic approach for designing contact matrices suitable for our problem setting, along with a simple decoding algorithm for reconstruction. Our approach is to first introduce a rather different setting for the problem that involves no randomness in the way the defective items become active. Namely, in the new setting an adversary can arbitrarily decide whether a certain contact with a defective item results in a positive test result or not, and the only restriction on the adversary is on the total amount of inactive contacts being made. The reason for introducing the adversarial problem is its combinatorial nature that allows us to use standard tools and techniques already developed in combinatorial group testing. Fortunately, it turns out that by choosing a carefully-designed value for the total amount of inactive contacts based on the parameters of the system, solving the adversarial variation is sufficient for the original (stochastic) problem.

Our task is then to design contact matrices suitable for the adversarial problem. We give a probabilistic construction of the contact matrix in Section~\ref{chap6sec:probDesign}. The probabilistic construction requires each test to independently contact the items with a certain well-chosen probability. This construction ensures that the resulting data gathered at the end of the experiment can be used for correct identification of the defective items with overwhelming probability, provided that the number of tests is sufficiently large.  In our analysis, we consider two different design strategies

\begin{itemize}
  \item \textbf{Per-Instance Design}: The contact matrix is suitable for every arbitrary, but a priori fixed, sparse input vector with overwhelming probability.
  \item \textbf{Universal Design}: The contact matrix is suitable for all sparse input vectors with overwhelming probability.
\end{itemize}
Based on the above definitions, the contact matrix constructed for the per-instance scenario, once fixed, may fail to distinguish between all pairs of sparse input vectors. On the other hand, in the universal design, one can use a single contact matrix to successfully measure all sparse input vectors with a very high probability of success. Our results show that $M = O\left(K \log (N) / p^3\right)$ tests are sufficient for successful recovery in the per-instance scenario while we need $M = O\left(K^2 \log (N/K) / p^3\right)$ tests for the universal design.

\vspace{2mm} {\it Remark:} As is customary in the standard group testing literature, we think of the sparsity $K$ as a parameter that is noticeably smaller than the population size $N$, for example, one may take $K=O(N^{1/4})$.  Indeed, if $K$ becomes comparable to $N$, there would be little point in using a group testing scheme and in practice, for large $K$ it is generally more favorable to perform trivial tests on the items.
\end{section}

\begin{section}{Adversarial Setting}
\label{chap6:adversarial}
The problem described in Section~\ref{chap6:problemDefinition} has a stochastic nature, i.e., the sampling matrix is obtained from the contact matrix through a random process. In this section, we introduce an adversarial variation of the problem whose solution leads us to the solution for the original stochastic problem.

In the adversarial setting, the sampling matrix is obtained from the contact matrix by flipping up to $e$ arbitrary entries to $0$ on the support (i.e., the set of nonzero entries) of each column of $\cMc$.  The goal is to be able to exactly identify the sparse input vector despite the perturbation of the contact matrix and regardless of the choice of the flipped entries.  Note that the classical group testing problem corresponds to the special case $e=0$. Thus, the only difference between the adversarial problem and the stochastic one is that in the former, the flipped entries of the contact matrix are chosen arbitrarily (as long as there are not too many flips) while in the latter, they are chosen according to a specific random process.

It turns out that the combinatorial tool required for solving the adversarial problem is closely related to the notion of \emph{disjunct} matrices that is well studied in the group testing literature~\cite{DuH99}. The formal definition is as follows.

\begin{definition}
\label{chap6def:disjunct} A boolean matrix $\cM$ with $N$ columns $\cM_1, \ldots, \cM_{\!N}$ is called $(K,e)$-disjunct if, for every subset $S$ of the columns with $|S| \leq K$, and every $i \in [N]$, we have
\begin{equation*}
  \left|\supp(\cM_i) \setminus \left(\bigcup_{j \in S \setminus \{i\}} \supp(\cM_j) \right)\right| > e
\end{equation*}
where $\supp(\cM_i)$ denotes the support set of the column $\cM_i$ and $\setminus$ is the set difference operator. In words, this operation counts the number of nonzero positions in the column $\cM_i$ for which all columns with index in the set $S$ have zeros.
\end{definition}

Note that the special case of $(K, 0)$-disjunct matrices corresponds to the classical notion of $K$-disjunct matrices which is essentially equivalent to strongly selective families and superimposed codes (see \cite{ref:CMS01}). Moreover, when all columns of the matrix have the same Hamming weight $t$, a $(K+1, 2t/3)$-disjunct matrix turns out to be equivalent to $t$-majority $K$-strongly selective families that are defined in \cite{ref:Iwe08} (where each row of the matrix defines the characteristic vector of a set in the family). This notion is known to be useful for construction of non-adaptive compressed sensing schemes \cite{ref:Iwe08}.

The following proposition shows the relationship between contact matrices suitable for the adversarial problem and disjunct matrices.

\begin{proposition}
\label{chap6:disjunct}
Let $\cM$ be a $(K, e)$-disjunct matrix. Then taking $\cM$ as the contact matrix solves the adversarial problem for $K$-sparse vectors with error parameter $e$. Conversely, any matrix that solves the adversarial problem must be $(K-1, e)$-disjunct.
\end{proposition}

\begin{proof}
Let $\cM$ be a $(K, e)$-disjunct matrix and consider $K$-sparse vectors $\vx$ and $\vx'$ supported on different subsets $S$ and $S'$, respectively. Take an element $i \in S'$ which is not in $S$. By Definition~\ref{chap6def:disjunct}, we know that the column $\cM_i$ has more than $e$ entries on its support that are not present in the support of any $\cM_j, \:j \in S$.  Therefore, even after $e$ bit flips in $\cM_i$, at least one entry in its support remains that is not present in the test outcome of $\vx'$, and this makes $\vx$ and $\vx'$ distinguishable.

For the reverse direction, suppose that $\cM$ is not $(K-1,e)$-disjunct and take any $i$ and a subset $S$ with $|S| \leq K-1$, $i \notin S$ which demonstrates a counterexample for $\cM$ being $(K-1, e)$-disjunct.  Consider $K$-sparse vectors $\vx$ and $\vx'$ supported on $S$ and $S \cup \{i\}$, respectively. An adversary can flip up to $e$ bits on the support of $\cM_i$ from $1$ to $0$, leave the rest of $\cM$ unchanged, and ensure that the test outcomes for $\vx$ and $\vx'$ coincide. Thus $\cM$ is not suitable for the adversarial problem.
\end{proof}

\begin{subsection}{Distance Decoder}
Proposition~\ref{chap6:disjunct} shows that a $(K, e)$-disjunct contact matrix can \emph{combinatorially} distinguish between $K$-sparse vectors in the adversarial setting with error parameter $e$. In the following, we show that there exists a much simpler decoder for this purpose.

\noindent \textbf{\\Distance decoder:} For any column $\vc_i$ of the contact matrix $\cMc$, the decoder verifies the following:
\begin{equation}
\label{chap6:decoder}
  |\supp(\vc_i)\setminus\supp(\vy)|\leq e
\end{equation}
where $\vy$ is the vector consisting of the test outcomes. The coordinate $x_i$ is decided to be nonzero if and only if the inequality holds.

The above decoder is a straightforward generalization of a standard decoder that is used in classical group testing. The standard decoder chooses those columns of the measurement matrix whose supports are fully contained in the measurement outcomes (see \cite[Ch~7]{DuH99}).

\begin{proposition}
\label{chap6:distanceDecoder}
Let the contact matrix $\cMc$ be $(K, e)$-disjunct. Then, the distance decoder correctly identifies the correct support of any $K$-sparse vector in the adversarial setting with error parameter $e$.
\end{proposition}

\begin{proof}
Let $\vx$ be a $K$-sparse vector and $S := \supp(\vx)$, $|S| \leq K$, and $\cMc_S$ denote the corresponding set of columns in the sampling matrix.  Obviously, all the columns in $\cMc_S$ satisfy~\eqref{chap6:decoder} (as no column is perturbed in more than $e$ positions) and thus the reconstruction includes the support of $\vx$ (this is true regardless of the disjunctness property of $\cMc$). Now let the vector $\hat{\vy}$ be the bitwise OR of the columns in $\cMc_S$ and assume that there is a column $\vc$ of $\cMc$ outside $S$ that satisfies~\eqref{chap6:decoder}. Thus, since
\begin{equation*}
\supp(\vy) \subseteq \supp(\hat{\vy})\,,
\end{equation*}
we will have
\begin{equation*}
|\supp(\vc)\setminus\supp(\hat{\vy})|\leq e
\end{equation*}
which violates the assumption that $\cMc$ is $(K, e)$-disjunct for the support set $S$ and the column $\vc$ outside this set. Therefore, the distance decoder outputs the exact support of $\vx$.
\end{proof}

\end{subsection}

Of course, posing the adversarial problem is interesting if it helps in solving the original stochastic problem from which it originates. In the next section,  we show that this is indeed the case; and in fact the task of solving the stochastic problem reduces to that of the adversarial problem.

\end{section}

\begin{section}{Probabilistic Design}
\label{chap6sec:probDesign}
In this section, we consider a probabilistic construction for $\cMc$, where each entry of $\cMc$ is set to $1$ independently with probability $q = \alpha/K$, for a parameter $\alpha$ to be determined later, and $0$ with probability $1-q$. We will use standard arguments to show that, if the number of tests $M$ is sufficiently large, then the resulting matrix $\cMc$ is suitable with all but a vanishing probability.

By looking carefully at the proof of Proposition~\ref{chap6:distanceDecoder}, we see that there are two events that \emph{may prevent} the distance decoder with error parameter $e$ to successfully recover the input vector $\vx$ with support on $S$:
\begin{enumerate}
  \item There are more than $e$ flips on the columns of the contact matrix in $S$.
  \item There exists a column outside $S$ where the $(K,e)$-disjunct property is violated.
\end{enumerate}

Based on these observations, the number of tests required for building suitable contact matrices are given by the following theorem.

\begin{theorem}
\label{chap6:findM}
Consider $M \times N$ contact matrices $\cMc$ constructed by the probabilistic design procedure. If $M = O\left(K \log (N) / p^3\right)$ for the per-instance scenario or $M = O\left(K^2 \log (N/K) / p^3\right)$ for the universal scenario, then the probability of failure for the reconstruction with the distance decoder goes to zero as $N \rightarrow \infty$.
\end{theorem}

\begin{proof}
Let $e$ be the decision-making parameter of the distance decoder. We first find an upperbound for the number of bit flips in any column of the contact matrix $\cMc$. To this end, take any column $\cMc_i$ of $\cMc$. Each entry of the column $\cMc_i$ is flipped independently with probability $(1-p)q$ which, on average, results in $(1-p)qM$ bit flips per column. Let $e = (1+\delta)(1-p)\,qM$ for a constant $\delta > 0$. By Chernoff bounds (cf.\ \cite{ref:MU}), the probability that the amount of bit flips exceeds $e$ is at most
\begin{equation}
\label{cha6equ:bitflipupperbound}
  \exp\left(-\delta^2 (1-p)qM/(2+\delta)\right)\,.
\end{equation}

Second, we check the disjunctness property of the contact matrix for this parameter $e$. To this end, consider any set $S$ of $K$ columns of $\cMc$, and any column outside these, say the $i$th column where $i \notin S$. First we upper bound the probability of \emph{failure} for this choice of $S$ and $i$. That is, the probability that the number of rows that have a $1$ at the $i$th column and all-zeros at the positions corresponding to $S$ is at most $e$. Clearly if this event happens the $(K, e)$-disjunct property is violated.

A row is \emph{good} if at that row the $i$th column has a $1$ but all the columns in $S$ have zeros. For a particular row, the probability that the row is good is $q(1-q)^K$ (using independence of the entries of the measurement matrix). Then failure corresponds to the event that the number of good rows is at most $e$. The distribution on the number of good rows is binomial with mean $\mu = q(1-q)^K M$. Using the Chernoff bound and assuming that $e < \mu$ (we will choose $\alpha$ and $\delta$ to ensure this condition is satisfied), we have\footnote{The failure probability is at least $0.5$ if $e \geq \mu\,.$}
\begin{eqnarray*}
  \text{failure probability} &\leq& \exp\left( \frac{-(\mu-e)^2}{2\mu}\right) 
\end{eqnarray*}
\begin{eqnarray}
\label{chap6equ:failureProb2}
  &=& \exp\left(\frac{-M q \left((1-q)^K - (1-p)(1+\delta)\right)^2}{2(1-q)^K}\right) \\
  &\overset{(a)}{\leq}& 
  \exp\left(\frac{-M q \left((1-2\alpha) - (1-p)(1+\delta)\right)^2}{2(1-\alpha/2)}\right) \nonumber \\
  &\leq & \exp\left(-\frac{1}{2} M q \left((1-2\alpha) - (1-p)(1+\delta)\right)^2\right) \nonumber \\
  &=& \exp(-M \gamma / K) \label{eqn:MgammaK},
\end{eqnarray}
where we have defined
\[
\gamma := \frac{\alpha}{2} \left((1-2\alpha) - (1-p)(1+\delta)\right)^2
\]
The inequality $(a)$ is due to the fact that $(1-q)^K = (1-\alpha/K)^K$ is always between $3^{-\alpha}$ and $2^{-\alpha}$, and in particular for $\alpha \in [0,1]$, this range is strictly contained in $[1-2\alpha, 1-\alpha/2]$. Note that by choosing the parameters $\alpha$ and $\delta$ sufficiently small, the quantity 
\[
\left((1-2\alpha) - (1-p)(1+\delta)\right)^2
\]
 in the exponent can be made arbitrarily close to $p^2$. As a concrete choice, however, we take $\delta := p/2$ and $\alpha := p/8$ which gives \[\gamma = \frac{p}{16} \left( \frac{p}{4} + \frac{p^3}{4}\right)^2 = \Omega(p^3)\]
 and, therefore,
\begin{equation} \label{eqn:failure}
\text{failure probability} \leq 2^{-\Omega(Mp^3/K)}.
\end{equation}
In order to calculate the number of tests, we consider per-instance and universal scenarios separately.
\begin{itemize}
  \item \emph{Per-Instance Scenario}\\
  For the per-instance scenario, the disjunctness property needs to hold only for a fixed set $S$, corresponding to the support of the fixed sparse vector that defines the instance. Therefore, we only need to apply the union bound over all possible choices of $i$ for a fixed set $S$. From \eqref{eqn:failure}, the probability of coming up with a bad choice of $\cMc$ would thus be at most \[ N 2^{-\Omega(Mp^3/K)}. \] This probability vanishes for an appropriate choice of 
\begin{equation}
  M = \Theta\left( \frac{K \log N}{p^3} \right). \label{eqn:choiceOfM}
\end{equation}
At the same time, using~\eqref{cha6equ:bitflipupperbound} and the union bound, the probability that the amount of bit flips in any of the $K$ columns in $S$ exceeds $e$ is upper bounded by
\begin{multline*}
  K \exp\left(-\delta^2 (1-p)qM/(2+\delta)\right) \\= K \exp\left(-\Omega\left(\frac{\alpha \delta^2 (1-p)\log (N)}{(2+\delta)p^3}\right)\right),
\end{multline*}
which is vanishing (i.e., $o(1)$) assuming the constant behind the $\Theta(\cdot)$ notion in \eqref{eqn:choiceOfM} is sufficiently large.
Therefore, the distance decoder successfully decodes the input vector $\vx$ with probability $1-o(1)$ in the per-instance scenario with $M = O\left(K \log (N) / p^3\right)$ tests.

\item \emph{Universal Scenario}\\
In this case, we apply the union bound over all possible choices of $S$ and $i$. Using \eqref{eqn:MgammaK}, the probability of coming up with a bad choice of $\cMc$ is at most $N \binom{N}{K} \exp\left(-M \gamma/K\right)\,$. This probability vanishes for an appropriate choice of
\begin{equation*}
  M = \Theta \left( \frac{K^2 \log (N/K)}{\gamma} \right) = \Theta\left( \frac{K^2 \log (N/K)}{p^3} \right).
\end{equation*}
At the same time, using~\eqref{cha6equ:bitflipupperbound} and the union bound, the probability that the amount of bit flips in any of the $N$ columns of the contact matrix exceeds $e$ is upper bounded by
\begin{multline*}
  N \exp\left(-\delta^2 (1-p)qM/(2+\delta)\right) \\= N \exp\left(-\frac{\alpha \delta^2 (1-p)K\log (N/K)}{(2+\delta)\gamma}\right) 
  = o(1).
\end{multline*}
Therefore, with $M = O\left(K^2 \log (N/K) / p^3\right)$ tests, the probabilistic design constructs a contact matrix such that the distance decoder is able to decode all sparse input vectors $\vx$ with probability $1-o(1)$.

\end{itemize}
\end{proof}

The probabilistic construction results in a rather sparse contact matrix, namely, one with density $O(1/K)$ that decays with the sparsity parameter $K$. In the following, we show that sparsity is necessary for the probabilistic construction to work.

\begin{proposition}
Let $\cM$ be an $M \times N$ boolean random matrix, where $M = O\left(K^2 \log (N/K)\right)$ or $M = O\left(K \log (N)\right)$ for an integer $K > 0$, which is constructed by setting each entry independently to $1$ with probability $q$. Then either $q = O\left(\log K/K\right)$ or otherwise the probability that $\cM$ is $(K,e)$-disjunct (for any $e \geq 0$) approaches to zero as $N$ grows.
\end{proposition}

\begin{proof}
Suppose that $\cM$ is an $M \times N$ matrix that is $(K,e)$-disjunct. Observe that, for any integer $t \in (0,K)$, if we remove any $t$ columns of $\cM$ and all the rows on the support of those columns, the matrix must remain $(K-t, e)$-disjunct. This is because any counterexample for the modified matrix being $(K-t, e)$-disjunct can be extended to a counterexample for $\cM$ being $(K,e)$-disjunct by adding back the removed columns and rows.

Now consider any $t$ columns of $\cM$, and denote by $M_0$ the number of rows of $\cM$ at which the entries corresponding to the chosen columns are all zeros. The expected value of $M_0$ is $(1-q)^t M$. Moreover, for every $\delta \in (0,1\!)$ we have
\begin{equation}
\label{chap6equ:chernoffDisj}
  \Pr\left[M_0 > (1+\delta) (1-q)^t M\right] \leq \exp\left( -\frac{\delta^2}{3} (1-q)^t M\right)
\end{equation}
by the Chernoff bound. Let $t_0$ be the largest integer for which
\begin{equation*}
(1+\delta) (1-q)^{t_0} M \geq \log N\,.
\end{equation*}
If $t_0 < K-1$, we let $t := t_0 + 1$ above, and this makes the right hand side of~\eqref{chap6equ:chernoffDisj} upper bounded by $o(1)$. So with probability $1-o(1)$, the chosen $t$ columns of $\cM$ will keep $M_0$ at most $(1+\delta)(1-q)^t M$. Removing those columns and all the rows on the support of these columns leaves the matrix $(K-t, e)$-disjunct, which obviously requires at least $\log N$ rows as even a $(1, 0)$-disjunct matrix needs so many rows. Therefore, we must have
\begin{equation*}
  (1+\delta)(1-q)^t M \geq \log N\,.
\end{equation*}
However, this inequality is not satisfied by the assumption on $t_0$.  So if $t_0 < K-1$, little chance remains for $\cM$ to be $(K,e)$-disjunct for any $e \geq 0$. Therefore, we should have $t_0 \geq K-1$. Using the condition on $t_0$, we have
\begin{equation*}
  (1+\delta)(1-q)^{K-1} M \geq \log N\,.
\end{equation*}
This is equivalent to
\begin{equation*}
  q \leq \frac{\log\left(M(1+\delta)/\log N\right)}{K-1}
\end{equation*}
which for $M = O\left(K^2 \log (N/K)\right)$ or $M = O\left(K \log (N)\right)$ gives $q = O\left(\log K/K\right)$\,.
\end{proof}
In summary, our results indicate that for both the per-instance and universal settings, the activation probability $p$ increases the upper bound on the number of tests by a factor of $1/p^3$. Moreover, we can use the simple distance decoder to recover the unknown input vector with the complexity of $O(MN)$. However, in order for the probabilistic design to work, we should choose a flip probability $q$ such that $q = O\left(\log K/K\right)$. In fact, our choice of $q = \alpha/K$ for a constant $\alpha$ satisfies this requirement.

\end{section}

\begin{section}{System Design and Simulation Results}
\label{chap6:design_simulation}
In this section, we provide a systematic design procedure which gives us the number of tests necessary for the decoding process to be successful. While the design procedure applies to both per-instance and universal scenarios, the numerical simulation result is provided only for the per-instance setting, since evaluating the universal design requires to test all possible inputs which is computationally prohibitive.

According to the discussion in Section~\ref{chap6sec:probDesign}, there are two types of failure events which we want to avoid in designing the contact matrix $\cMc$. The first failure event, denoted as $f_1$, happens when the number of bit flips in a column is not tolerable by the contact matrix and the second event, denoted as $f_2$, relates to the violation of the disjunctness property of the matrix. The inputs to the design procedure are $N$, $K$, $p$, $p_{f_1}^{ }$ and $p_{f_2}^{ }$, where the last two parameters denote the maximum tolerable probability for the first and second failure events, respectively. Then, the design procedure should provide us with the quantities $M$, $q$ and $e$, which are the required parameters to setup the sensing and recovery algorithms.

\noindent Let us summarize the results of the probabilistic design of Section~\ref{chap6sec:probDesign}. First we define $\eta$ from \eqref{chap6equ:failureProb2} as
\begin{equation}
\label{chap6equ:eta}
  \eta = q\frac{\left((1-q)^K - (1-p)(1+\delta)\right)^2}{2(1-q)^K}.
\end{equation}
Then, for the per-instance scenario (which is denoted by (i)), we have
\begin{multline*}
  p_{f_1}^{\text{(i)}} \leq 1-\\ \left[1-\exp\left(-(1-p)\,q\,M^{\text{(i)}}(\log(1+\delta)^{(1+\delta)}-\delta)\right)\right]^K, \\
\end{multline*}
\begin{align*}
  p_{f_2}^{\text{(i)}} &\leq N \exp\left(-M^{\text{(i)}}\eta\right),\\
  e^{\text{(i)}} &= (1+\delta)(1-p)\,qM^{\text{(i)}}.
\end{align*}
For the universal strategy (which is denoted by (u)), we have
\begin{multline*}
  p_{f_1}^{\text{(u)}} \leq 1-\\ \left[1-\exp\left(-(1-p)\,q\,M^{\text{(u)}}(\log(1+\delta)^{(1+\delta)}-\delta)\right)\right]^N,
\end{multline*}
\begin{align*}
  p_{f_2}^{\text{(u)}} &\leq N \binom{N}{K} \exp\left(-M^{\text{(u)}}\eta\right),\\
  e^{\text{(u)}} &= (1+\delta)(1-p)\,qM^{\text{(u)}}.
\end{align*}
Note that since the first failure event happens independently on the columns, we have used a more precise expression for this failure probability which does not use the union bound, and also makes use of the exact expression for the Chernoff bound.

Let us provide the details of the design for the per-instance scenario; The universal design follows the same lines. For any fixed value of $\alpha$, $\delta > 0$ is the only parameter which should be determined such that the failure probabilities fall below the maximum tolerable values. To this end, we initialize the value of $\delta$ to zero and increase it in small steps up to the value $\delta_{\text{max}}$. Given that $e = (1-p)(1+\delta)\,qM^{\text{(i)}}$  and $\mu = q(1-q)^K M^{\text{(i)}}$ and under the condition that $e < \mu$, we have
\begin{equation*}
  \delta_{\text{max}} = (1-q)^K/(1-p)-1\,.
\end{equation*}

For any value of $\delta < \delta_{\text{max}}$ and given the maximum tolerable probability for the second failure event $p_{f_2}^{ }$, the number of tests are computed as
\begin{equation*}
  M^{\text{(i)}} = \eta^{-1} \log{(N/p_{f_2}^{ })}
\end{equation*}
where $\eta$ is defined in~\eqref{chap6equ:eta}. This is then used to compute the corresponding probability for the first failure event as
\begin{multline*}
  p_{f_1}^{\text{(i)}} = 1-\\ \left[1-\exp\left(-(1-p)\,q\,M^{\text{(i)}}(\log(1+\delta)^{(1+\delta)}-\delta)\right)\right]^K.
\end{multline*}
We continue increasing $\delta$ until $p_{f_1}^{\text{(i)}}$ falls below the maximum tolerable probability for the first failure event $p_{f_1}^{ }$. This provides us with the number of tests $M^{\text{(i)}}$ and the error parameter $e^{\text{(i)}}$ for the chosen value of $\alpha$. This whole process is continued for different values of $\alpha$ in the range $[0, \alpha_{\text{max}}]$. At the end, we find the value of $\alpha$ which results in the minimum number of tests for the given value of $p$. This provides us with the number of tests and the decision parameter for the distance decoder.

\begin{figure*}[t]
\centering
\subfigure[Per-Instance Strategy]{
\includegraphics[width=.48\textwidth]{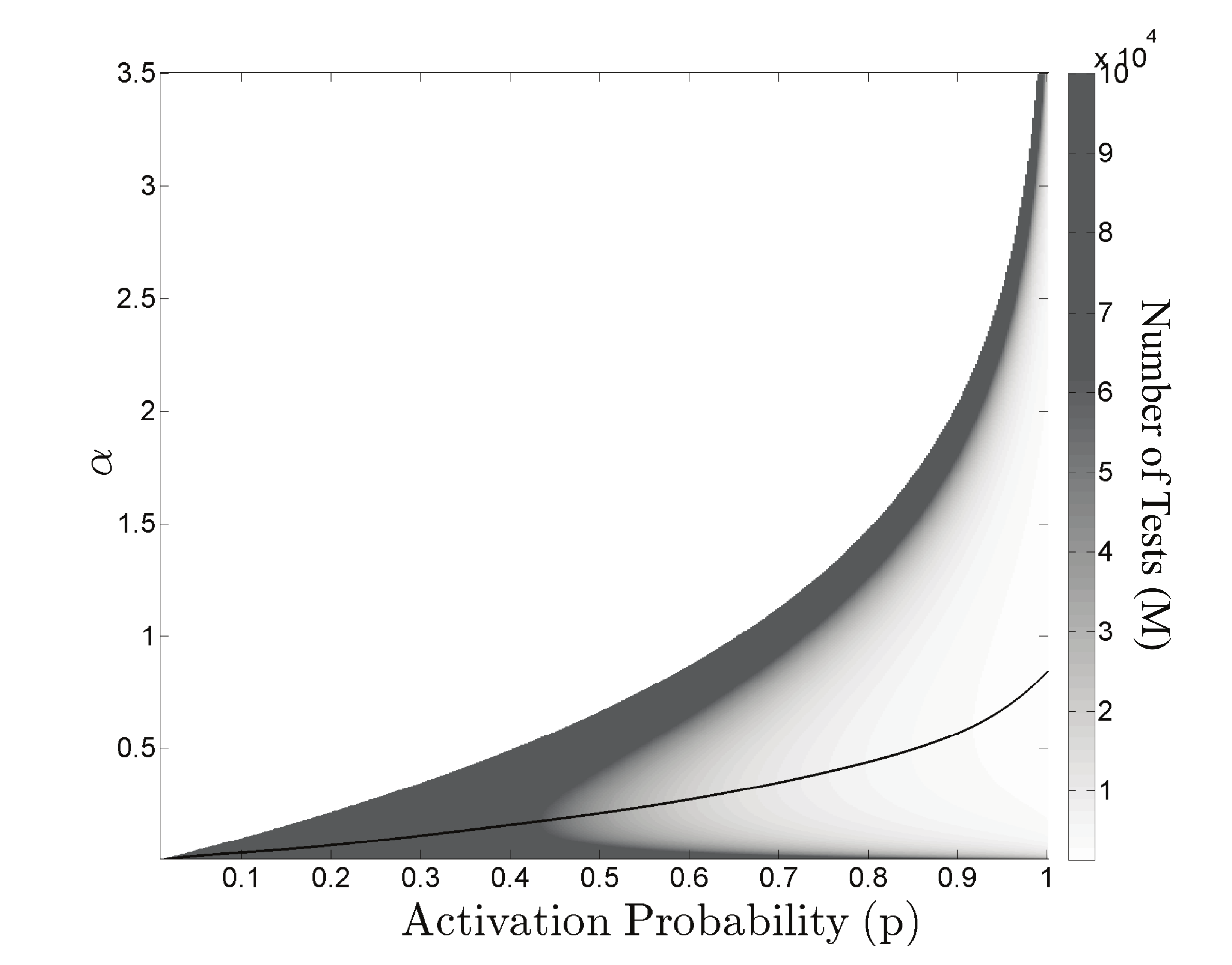}
\label{chap6subfig:designInstance}
}
\subfigure[Universal Strategy]{
\includegraphics[width=.48\textwidth]{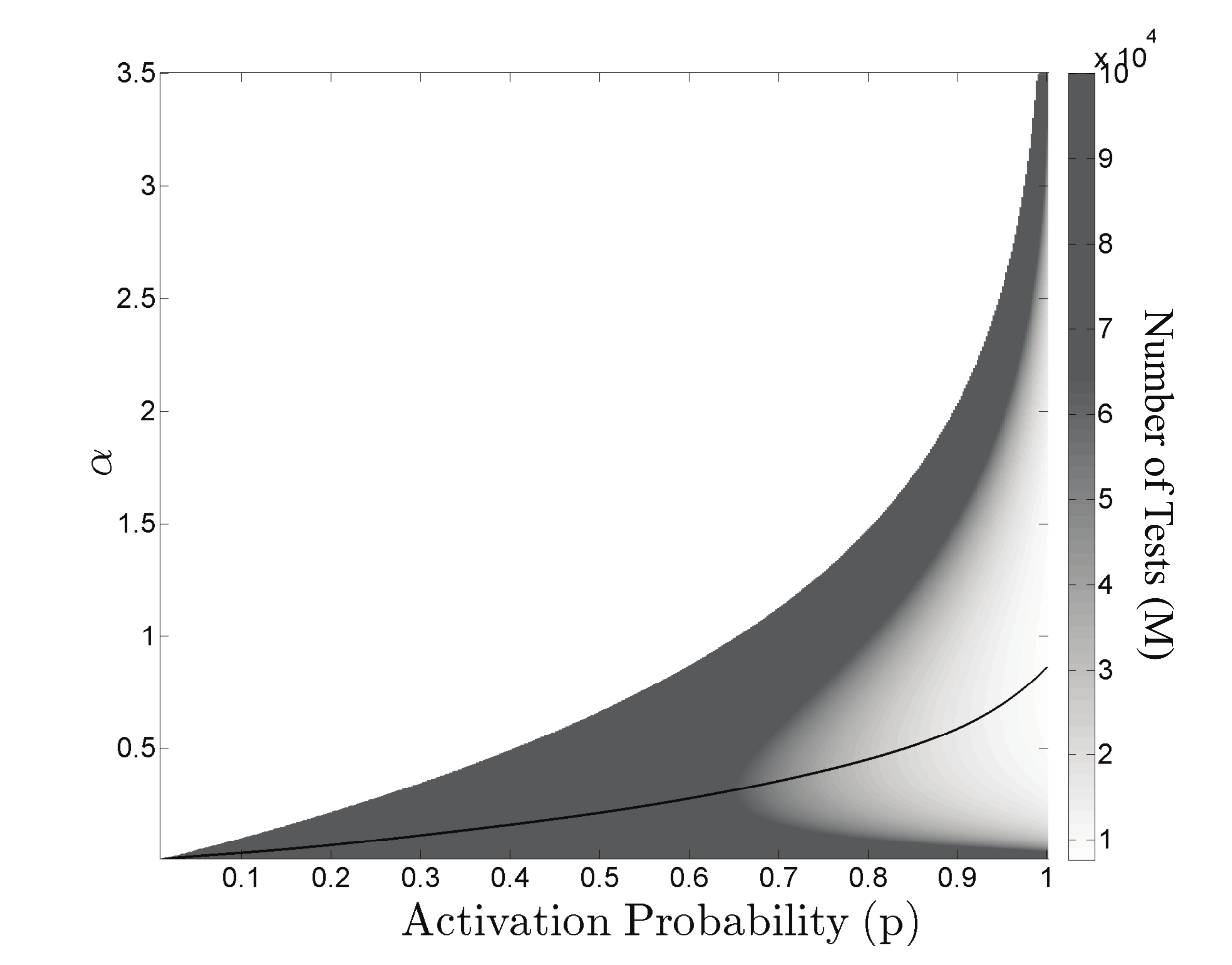}
\label{chap6subfig:designUniversal}
}
\subfigure[]{
\includegraphics[width=.48\textwidth]{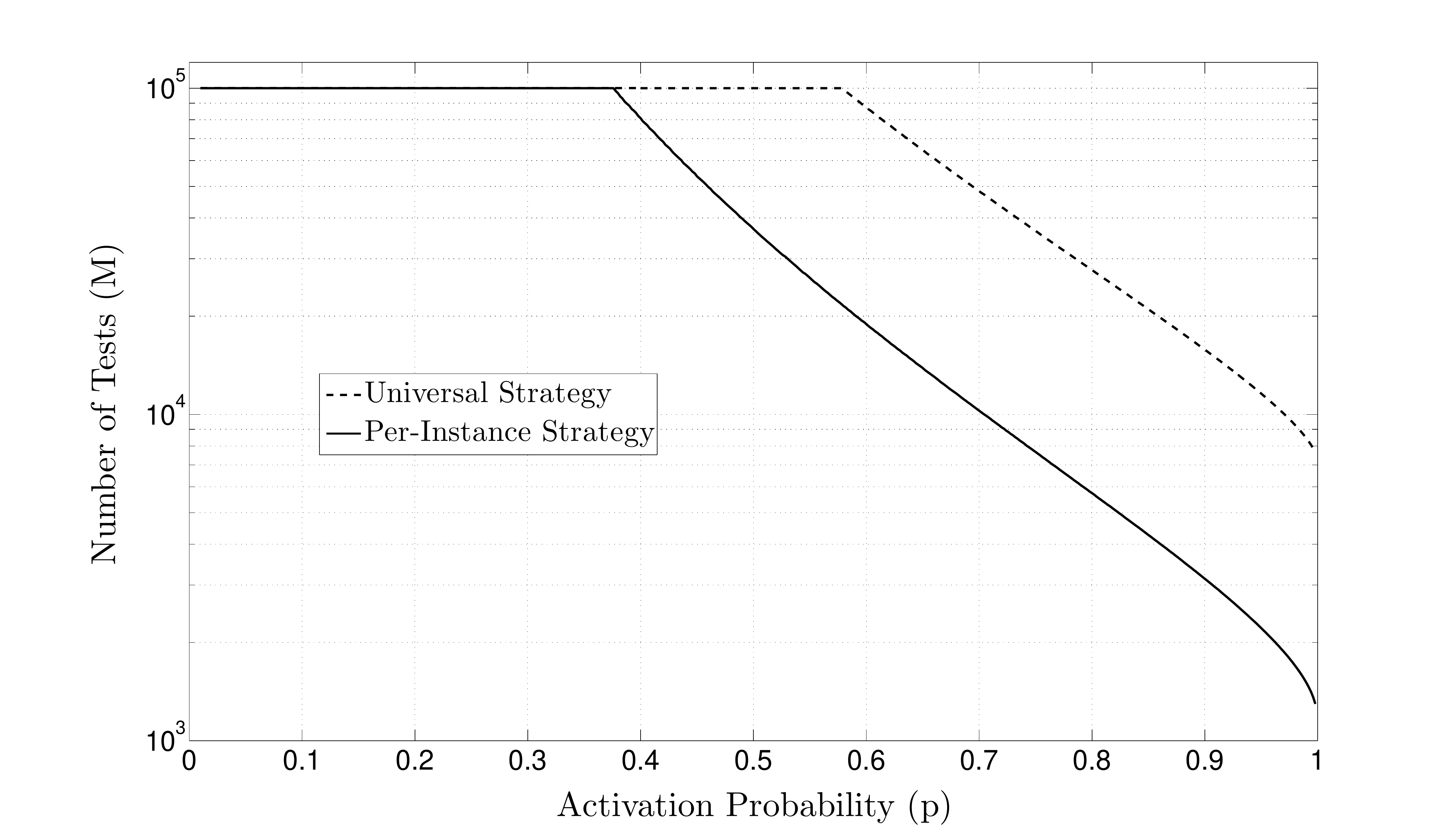}
\label{chap6subfig:blackcurves}
}
\caption{The number of tests ($M$) as the output of the design procedure for \subref{chap6subfig:designInstance} per-instance and \subref{chap6subfig:designUniversal} universal strategies, as a function of the parameter $\alpha$ and the activation probability $p$. The parameters are set to $N = 100'000$, $K = 10$ and $p_{f_1}^{ } = p_{f_2}^{ } = 0.001\,$. The black curves provide us with the value of $\alpha$ which gives the minimum number of tests for each value of the activation probability $p$. \subref{chap6subfig:blackcurves} The minimum number of tests (corresponding to the black curves in \subref{chap6subfig:designInstance} and \subref{chap6subfig:designUniversal}) for universal and per-instance strategies, as a function of the activation probability $p$.}
\end{figure*}

\begin{figure*}[t]
\centering
\subfigure[Per-Instance Strategy]{
\includegraphics[width=0.48\textwidth]{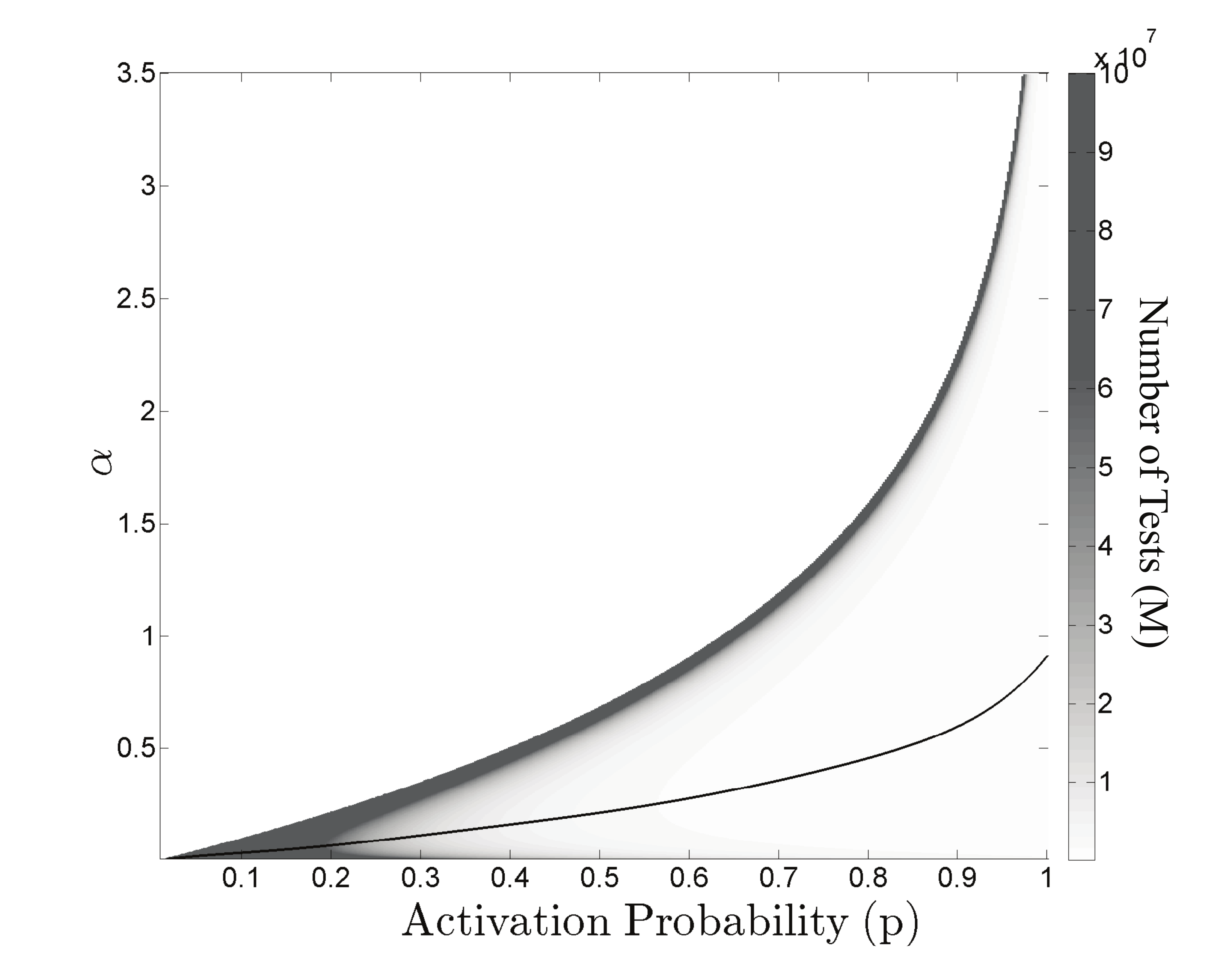}
\label{chap6subfig:designInstance2}
}
\subfigure[Universal Strategy]{
\includegraphics[width=0.48\textwidth]{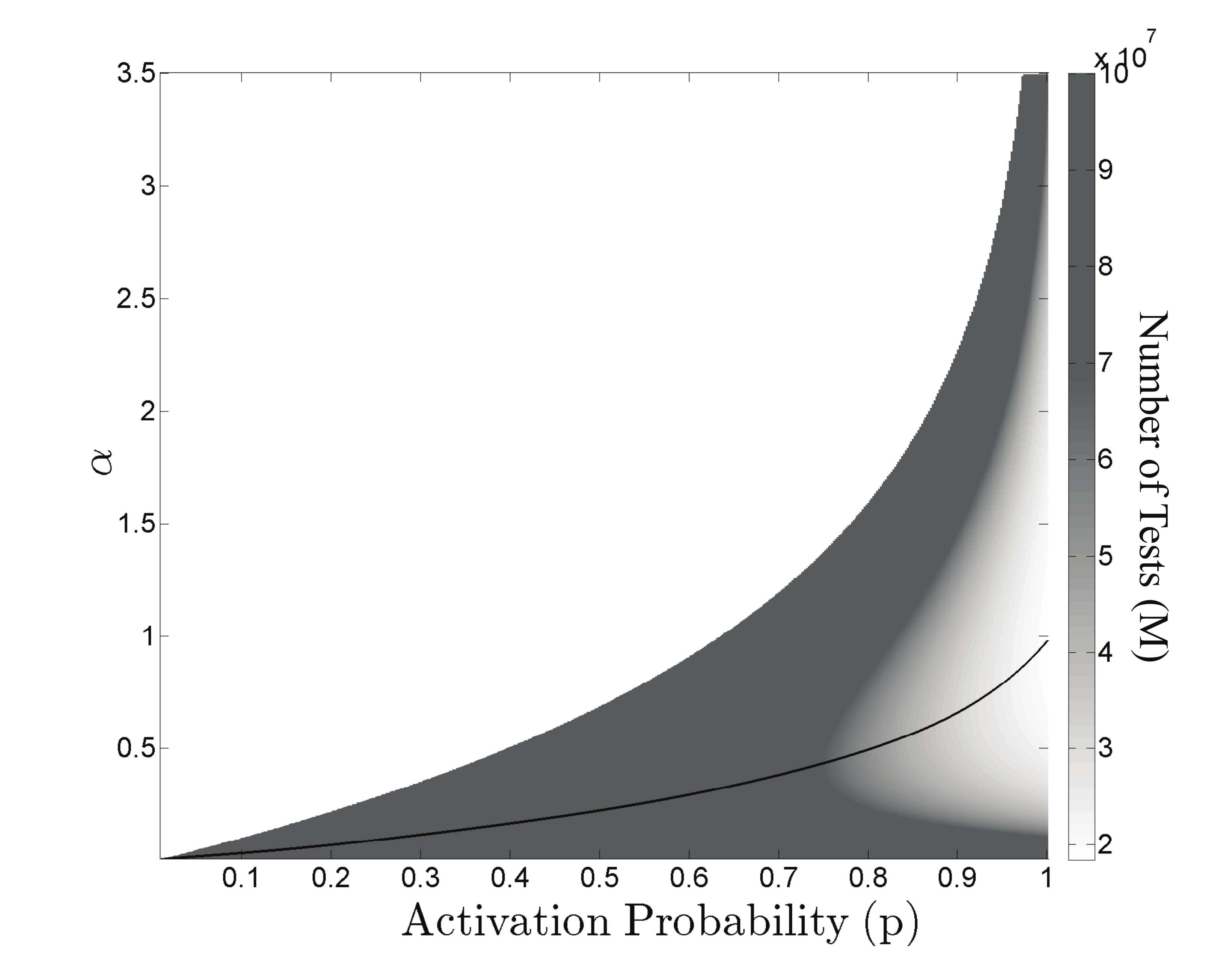}
\label{chap6subfig:designUniversal2}
}
\subfigure[]{
\includegraphics[width=0.48\textwidth]{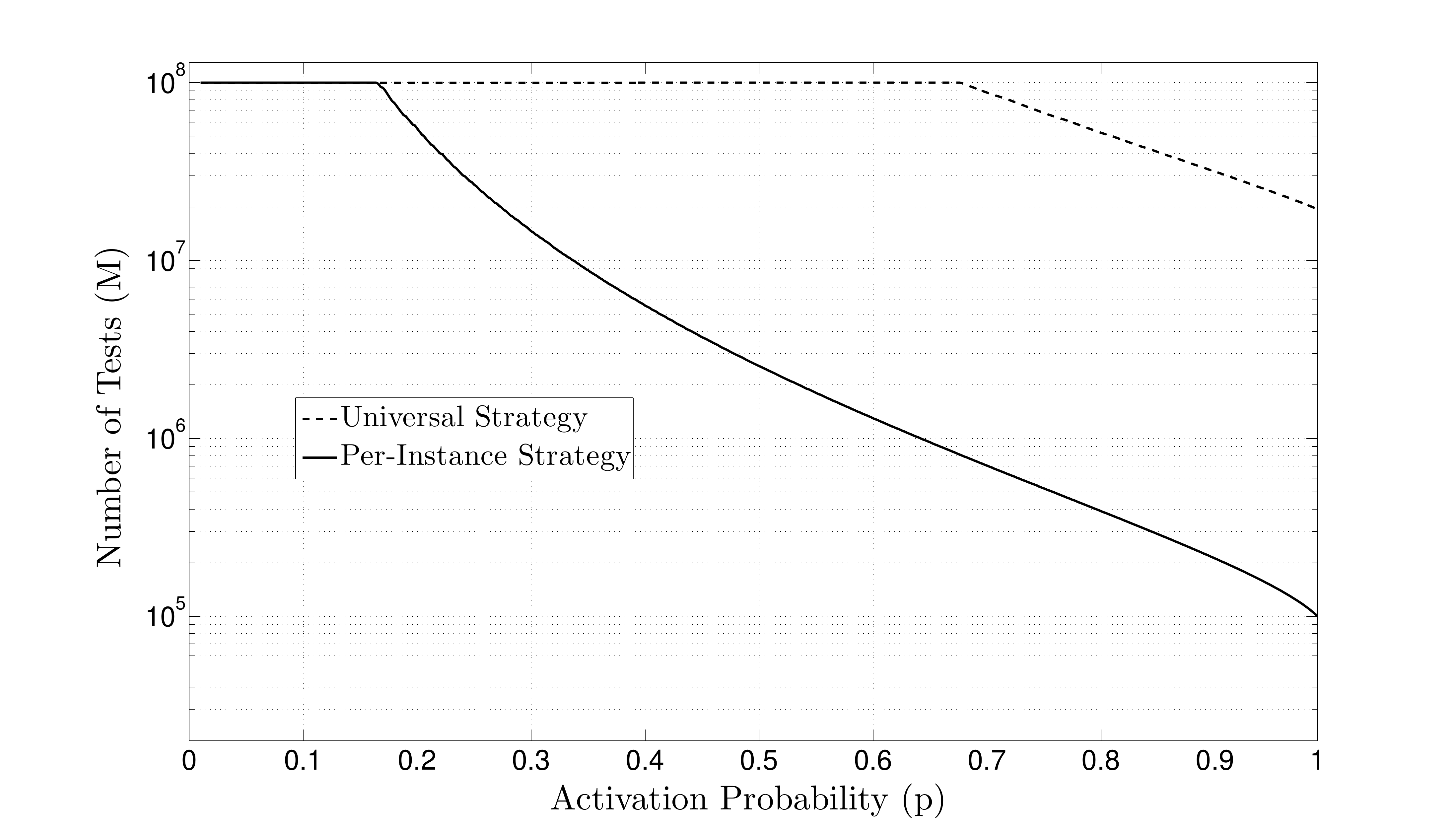}
\label{chap6subfig:blackcurves2}
}
\caption{The number of tests ($M$) as the output of the design procedure for \subref{chap6subfig:designInstance2} per-instance and \subref{chap6subfig:designUniversal2} universal strategies, as a function of the parameter $\alpha$ and the activation probability $p$. The parameters are set to $N = 100'000'000$, $K = 500$ and $p_{f_1}^{ } = p_{f_2}^{ } = 0.001\,$. The black curves provide us with the value of $\alpha$ which gives the minimum number of tests for each value of the activation probability $p$. \subref{chap6subfig:blackcurves2} The minimum number of tests (corresponding to the black curves in \subref{chap6subfig:designInstance2} and \subref{chap6subfig:designUniversal2}) for universal and per-instance strategies, as a function of the activation probability $p$.}
\end{figure*}
Figures~\ref{chap6subfig:designInstance} and~\ref{chap6subfig:designUniversal} show the number of tests, for universal and per-instance strategies, as a function of the parameter $\alpha$ and the activation probability $p$. The population size is $N = 100'000$, the number of defective items is $K = 10$ and the maximum tolerable probabilities for the two failure events are set to $p_{f_1}^{ } = p_{f_2}^{ } = 0.001\,$.  Note that the number of tests for the per-instance scenario is much less than the universal scenario and moreover, it allows us to have designs appropriate for smaller activation probabilities. The black curve in each figure connects the points with minimum number of tests for each value of the activation probability $p$, which in turn provides us with the appropriate value for the parameter $\alpha$. The black curves are extracted and shown separately in Figure~\ref{chap6subfig:blackcurves}. In Figures~\ref{chap6subfig:designInstance2},~\ref{chap6subfig:designUniversal2} and~\ref{chap6subfig:blackcurves2},  we show the output of the design procedure for $N = 100'000'000$, $K = 500$ and $p_{f_1}^{ } = p_{f_2}^{ } = 0.001\,$.

In Figure~\ref{chap6subfig:measurementsVSfailure}, we set $N = 100'000$, $K = 10$ and $p = 0.8$ and use the design procedure to plot the number of tests as a function of the probability of failure in the per-instance strategy. Then, in Figure~\ref{chap6subfig:simulationResult}, we run a numerical experiment with the same values for the parameters $N$, $K$ and $p$ to assess the performance of the recovery algorithm, with the results averaged over $4000$ trials. We set the parameters $e$ and $\alpha$ for the numerical experiment equal to those which give us the probability of failure of $0.5$ in Figure~\ref{chap6subfig:measurementsVSfailure} (which are $e = 40$ and $\alpha = 0.44$) and change the number of tests. Note that although we expect a failure probability of around $0.5$ for $M=3000$ tests according to Figure~\ref{chap6subfig:measurementsVSfailure}, the recovery performance is much better in numerical simulations. This can be explained by noting that the upper bounds for the failure probabilities are not tight in general.

\begin{figure}[t]
\centering
\includegraphics[width=\columnwidth]{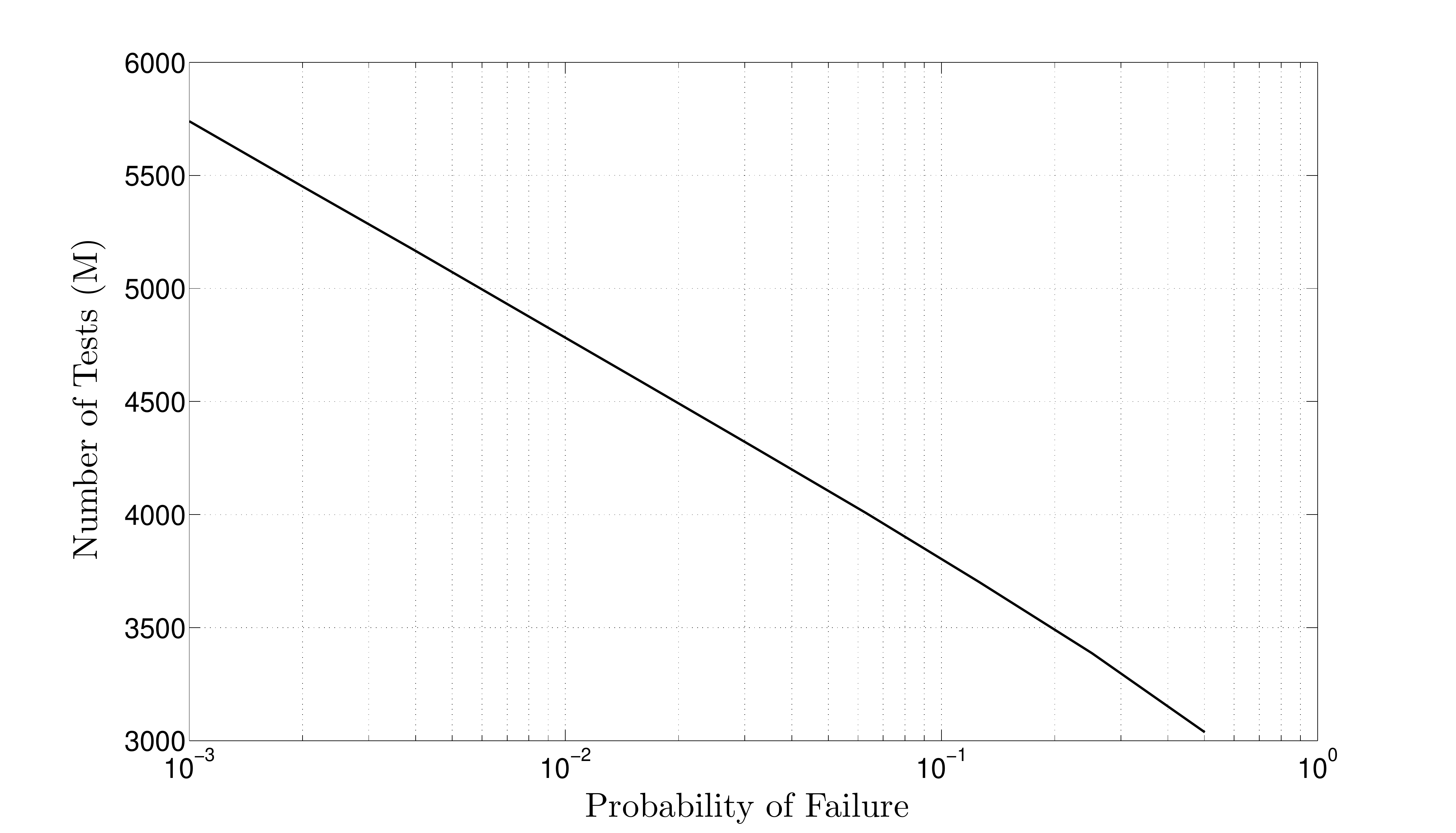}
\caption{The number of tests as a function of the failure probability for the per-instance scenario. The parameters are $N = 100'000$, $K = 10$ and activation probability $p = 0.8$.}
\label{chap6subfig:measurementsVSfailure}
\end{figure}

\begin{figure}
\centering
\includegraphics[width=\columnwidth]{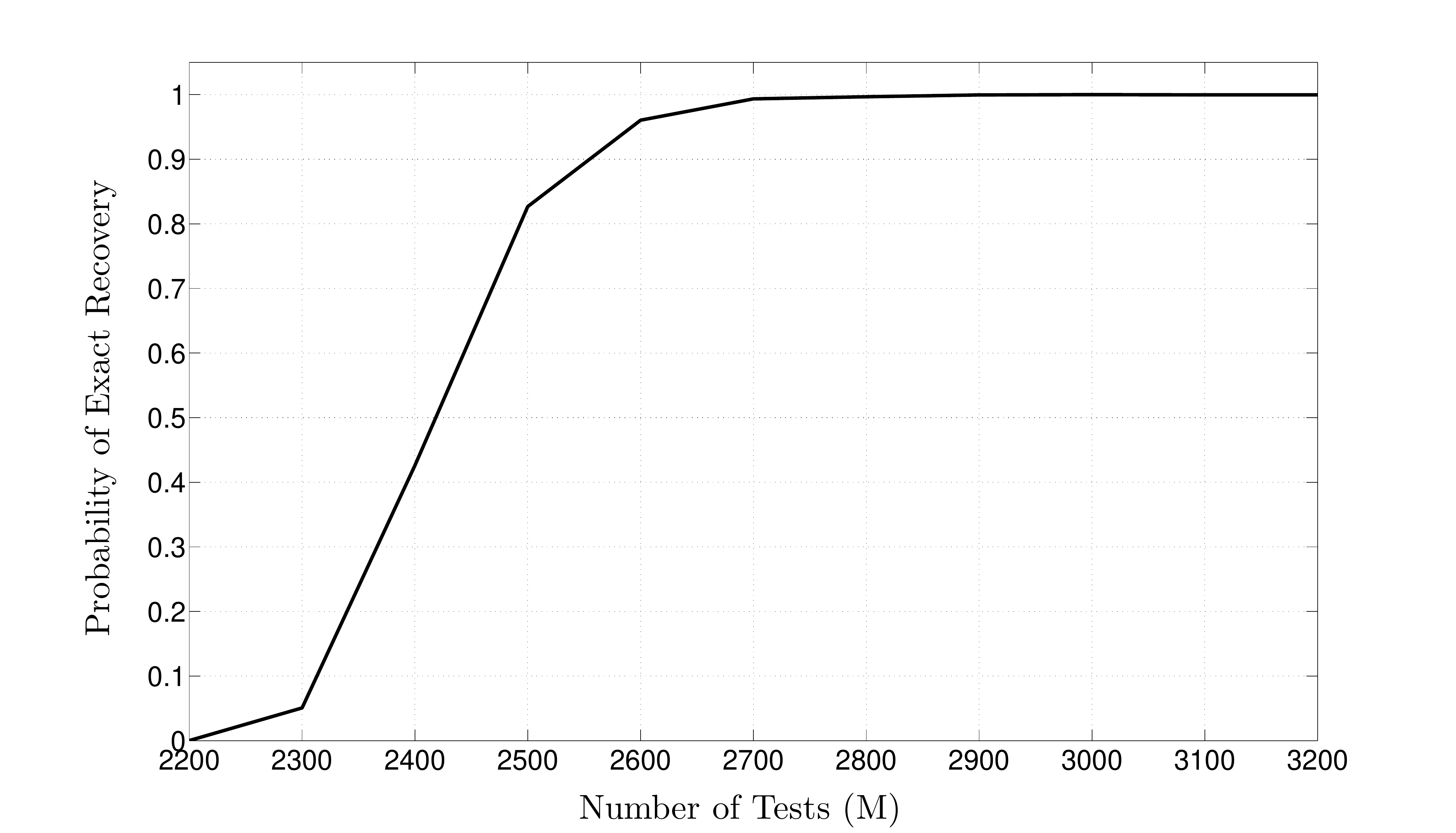}
\caption{The simulated probability of exact recovery for the per-instance strategy for $N = 100'000$, $K = 10$ and $p = 0.8$, averaged over $4000$ trials. Although the design procedure expects a high probability of failure for $M = 3000$ tests (see Figure~\ref{chap6subfig:measurementsVSfailure}), the recovery performance is much better in numerical simulations.}\label{chap6subfig:simulationResult}
\end{figure}

\end{section}

\begin{section}{Conclusion}
\label{chap6:summary}
We studied the problem of identifying a small number of defective items among a large population, using collective samples. With the viral epidemic application in mind, we investigated the case where the recovery algorithm may possess only partial knowledge about the sampling process, in the sense that the defective items can become inactive in the test results. We showed that by using a probabilistic model for the sampling process, one can design a non-adaptive contact matrix which leads to the successful identification of the defective items with overwhelming probability $1-o(1)$. We considered two strategies for the design procedure. In per-instance design, the contact matrix is suitable for each sparse input vector with overwhelming probability while in universal design, this is true for all sparse inputs. To this end, we proposed a probabilistic design procedure which requires a ``small'' number of tests to single out the sparse vector of defective items. More precisely, we showed that for an activation probability $p$, the number of tests sufficient for identification of up to $K$ defective items in a population of size $N$ is given by $M = O(K \log (N) / p^3)$ for the per-instance scenario and $M = O(K^2 \log (N/K) / p^3)$ for the universal scenario. Moreover, we proposed a simple decoder which is able to successfully identify the defective items with complexity of $O(MN)$. Finally, we provided a systematic design procedure which gives the number of tests $M$, along with the design parameters $\alpha$ and $e$, required for successful recovery. As expected, the numerical experiments showed that the number of tests provided by the design procedure overestimates the true one required to achieve a specified probability of failure. As a complement to this work, one can also consider the effects of false positives and false negatives on the required number of tests. We leave this issue for future work.
\end{section}


\bibliographystyle{ieeetr}\bibliography{references}
\vfill

\begin{IEEEbiographynophoto}{Mahdi Cheraghchi}
received the B.Sc.\ degree in computer engineering from Sharif University of Technology, Tehran, Iran, in 2004 and  the M.Sc.\ and Ph.D.\ degrees in computer science from EPFL, Lausanne, Switzerland, in 2005 and 2010, respectively. Since October 2010, he has been a post-doctoral researcher at the University of Texas at Austin, TX. His research interests include the interconnections between coding theory and theoretical computer science, derandomization theory and explicit constructions.
\end{IEEEbiographynophoto}

\begin{IEEEbiographynophoto}{Ali Hormati}
received the B.Sc.\ and M.Sc.\ degrees in communication systems from Sharif University of Technology, Tehran, Iran, in 2001 and 2003, respectively, and the Ph.D.\ degree in electrical engineering from EPFL, Lausanne, Switzerland in 2010.
From 2003 to 2006, he worked as a communication systems engineer in industry on the communication protocols for software radio systems. His collaboration with Qualcomm company during the Ph.D.\ studies resulted in patent filings on new techniques for MIMO channel estimation in OFDM and CDMA communication systems.  He is now a postdoctoral researcher in the Audiovisual Communications Laboratory at EPFL. His research interests include sampling theory, compressed sensing and inverse problems in acoustic tomography.
\end{IEEEbiographynophoto}

\begin{IEEEbiographynophoto}{Amin Karbasi}
received the B.Sc.\ degree in electrical engineering in 2004 and M.Sc.\ degree in communication systems in 2007 from EPFL , Lausanne, Switzerland. Since March 2008, he has been a Ph.D. student at EPFL. He was the recipient of the ICASSP 2011 best student paper award and ACM/Sigmetrics 2010 best student paper award. His research interests include graphical models, large scale networks, compressed sensing and information theory.
\end{IEEEbiographynophoto}

\begin{IEEEbiographynophoto}{Martin Vetterli}

received the Dipl.\ El.-‐Ing.\ degree from ETHZ, Switzerland, in 1981, the M.Sc.\ degree from Stanford University, CA, in 1982, and the Doctorat \`es Sciences degree from EPFL, Lausanne, Switzerland, in 1986.
He was a research assistant at Stanford and EPFL, and has worked for Siemens and AT\&T Bell Laboratories. In 1986 he joined Columbia University in New York, where he was last an Associate Professor of Electrical Engineering and co‐director of the Image and Advanced Television Laboratory. In 1993, he joined the University of California at Berkeley, where he was a Professor in the Department of Electrical Engineering and Computer Sciences until 1997, and now holds an Adjunct Professor position.
Since 1995 he is a Professor of Communication Systems at EPFL, Switzerland, where he chaired the Communications Systems Division (1996/97), and heads the Audiovisual Communications Laboratory. 
He has held visiting positions at ETHZ (1990) and Stanford (1998), 
is a fellow of IEEE, a fellow of ACM, a fellow of EURASIP, and a member of SIAM. 
He received the Best Paper Award of EURASIP in 1984, the Research Prize of the Brown Bovery Corporation (Switzerland) in 1986, the IEEE Signal Processing Society's Senior Paper Awards in 1991, in 1996 and in 2006 (for papers with D.\ LeGall, K.\ Ramchandran, and Marziliano and Blu, respectively). He won the Swiss National Latsis Prize in 1996, the SPIE Presidential Award in 1999, the IEEE Signal Processing Technical Achievement Award in 2001, and the IEEE Signal Processing Society Award in 2010.
He has published about 145 journal papers on a variety of topics in signal/image processing and communications, co-authored three books, and holds a dozen patents.
His research interests include sampling, wavelets, multirate signal processing, computational complexity, signal processing for communications, digital image/video processing, joint source/channel coding, signal processing for sensor networks and inverse problems like acoustic tomography.
\end{IEEEbiographynophoto}
\vfill

\end{document}